\documentclass[conference]{IEEEtran}
\IEEEoverridecommandlockouts
\usepackage{cite}
\usepackage{amsmath,amssymb,amsfonts}
\usepackage{algorithmic}
\usepackage{graphicx}
\usepackage{textcomp}
\usepackage{xcolor}
\usepackage[caption=false,font=footnotesize]{subfig}
\usepackage{amsthm}
\usepackage{pgfplots}
\pgfplotsset{compat=1.18}
\usetikzlibrary{tikzmark,arrows.meta,calc}
\newtheorem{Lemma}{Lemma}
\def\BibTeX{{\rm B\kern-.05em{\sc i\kern-.025em b}\kern-.08em
    T\kern-.1667em\lower.7ex\hbox{E}\kern-.125emX}}
\begin{document}

\title{Uncertainty-Aware Fusion for Resilient Distributed Radar Sensing\thanks{This work has been co-funded by the LOEWE
initiation (Hesse, Germany) within the
emergenCITY center
[LOEWE/1/12/519/03/05.001(0016)/72], and in part by the
German Federal Ministry for Research, Technology and Space (BMFTR) under
the program of “Souverän. Digital. Vernetzt.” joint project Open6GHub plus
(Project-ID 16KIS2407). }}

\author{Christian Eckrich, Maik Pfefferkorn, Rolf Findeisen,
Abdelhak M. Zoubir,
and Vahid Jamali
}

\maketitle

\begin{abstract}
Distributed radar sensing enables safe and resilient operation in mobile robotic and vehicular systems by combining multiple local views of the scene. In this paper, we consider a moving agent equipped with a frequency modulated continuous wave (FMCW) radar that fuses its local measurements with those of surrounding static radar sensors to reduce the uncertainty in target state estimation. The benefit of sharing measurements over capacity-limited wireless links come at the expense of two competing degradation mechanisms: quantization distortion, which increases the effective noise floor, and latency, which causes the received information to age. A unified framework is thus needed to quantify how individual radars contribute to uncertainty reduction under limited communication resources. To this end, we derive the Cram\'{e}r--Rao lower bound (CRLB) for the fused target state by combining local Fisher information matrices (FIMs) through coordinate transformations, incorporating distortion bounds from rate-distortion theory and information aging via a state transition model. We show that the resulting expression reveals a fundamental tradeoff between update rate and quantization fidelity governed by the available channel capacity. Our numerical simulations illustrate this tradeoff and demonstrate that a careful choice of system parameters (e.g., selected radar systems, quantization, update rate) is necessary for maximum uncertainty reduction.
\end{abstract}

\section{Introduction}
The operation of vehicles, drones and robots will increasingly rely on autonomous task-driven decision making and control under uncertainty \cite{Labib2021,Soliman2024,Fu2024}. In applications such as autonomous driving, drone swarms, and multi-robot systems in factories, agents operate and navigate in shared environments, where they have to maintain situational awareness of surrounding objects in order to assess potential collisions and other threats while performing their tasks \cite{Qamar2022}. 

To achieve this, the agents are usually equipped with local sensors, such as radar systems, cameras, or LiDARs, that provide onboard measurements related to their surroundings. However, it is not always guaranteed that the local measurements of an agent are sufficient to maintain situational awareness, e.g., due to occlusions, sensor limitations, or sensor failures, resulting in unacceptably high uncertainty in the state estimates of the object. 

In such cases, the agents may exploit the observations of surrounding static sensors to complement their local measurements and realize resilient situational awareness \cite{Moussa2024}. Retrieving the measurements of those static sensors requires communication resources, which are typically limited, especially in dynamic and unpredictable wireless environments \cite{Zou2023, Eckrich2024}. To cope with these communication constraints, the moving agent can select a subset of sensors which complement its local measurements the most. Quantization and increasing the update interval are used to further compress the data to be transmitted over the communication channel with limited capacity. Distortion noise due to quantization and increased age of information (AoI) due to larger update intervals degrade the quality of the shared measurements and thereby increase the uncertainty of the object state estimate at the agent. 
This paper provides a unified framework accounting for both quantization and latency as well as the value of individual static radars in the overall uncertainty reduction.
In particular, the contributions of this work are threefold:
\begin{itemize}
    \item A CRLB-based framework for fused object state estimation that captures the effects of quantization distortion and information aging due to latency.
    \item Analysis of the tradeoff between update rate and quantization fidelity under a channel capacity constraint.
    \item Numerical evaluation demonstrating the benefits of cooperative fusion, its sensitivity to quantization rate and latency and their underlying tradeoff.
\end{itemize}
In the subsequent sections, we will first derive the comprehensive CRLB expression for the fused object state estimation. Therefore, we formulate the local radar signal model in Section II, and then introduce the fused estimation model in Section III. Finally, we simplify and analyze the derived expression in Section IV, provide further evaluation numerically in Section V and conclude the paper in Section VI.

\section{Local Radar Signal Model}
Let us consider a moving agent $k = 0$ that is equipped with an FMCW radar sensor and observes a single point target\footnote{The analysis can be straightforwardly extended to multiple targets or the entire scene.}. The object parameters are collected in $\boldsymbol{\omega}_0 = [r_0,\theta_0,v_0]^\top$ and denote range, angle, and radial velocity of the object as seen by the moving agent, respectively. The measured echo signal is down-converted, sampled, and quantized to obtain the digital measurement cube, denoted as $\mathbf{Y}_0 \in \mathbb{C}^{N \times M \times L}$. Here, the dimensions of the data cube correspond to the fast-time, spatial, and slow-time dimensions of the radar measurement, respectively. For a single point target and a single propagation path, the data cube is modeled as
\begin{align}
    \label{eq:datacube}
    \mathbf{Y}_0
    &=
    \boldsymbol{M}_0(\boldsymbol{\omega}_0)
    + \mathbf{R}_0.
\end{align}
The $(n,m,l)$th entry of the signal component $\boldsymbol{M}_0(\boldsymbol{\omega}_0)\in \mathbb{C}^{N \times M \times L}$ within the data cube is given by
\begin{align}
    [\mathbf{M}_0(\boldsymbol{\omega}_0)]_{n,m,l} &= \sqrt{P_0}\,
    e^{-j(\Tilde{r}_0 n + \Tilde{\theta}_0 m + \Tilde{v}_0 l)}
\end{align}
with 
\begin{align}
    \Tilde{r}_0 = \frac{2B}{cN} r_0, \quad
    \Tilde{\theta}_0 = \frac{2\pi d}{\lambda M} \sin(\theta_0), \quad
    \Tilde{v}_0 = \frac{2T_c}{\lambda L} v_0,
\end{align}
where $B$ is the bandwidth, $c$ is the speed of light, $d$ is the antenna spacing, $\lambda$ is the wavelength, and $T_c$ is the chirp repetition interval. 
The term $\mathbf{R}_0\in \mathbb{C}^{N \times M \times L}$ denotes additive sensor noise, which is assumed to be zero mean complex Gaussian distributed with covariance matrix $\mathbf{\Sigma}_{R,0}$.
For the received power  model $P_0$, we aim to capture the key dependencies that relate to the object parameters, such as position ($\theta_0, r_0$) and radar cross section ($\sigma_\text{RCS}$), while summarizing the radar specific effects with a constant $\beta_0$. 
The received signal power is modeled using a simplified point-target radar equation 
\begin{align}
    \sqrt{P_0}
    =
    \frac{\beta_0 \cos^2(\theta_0)\sigma_\text{RCS}}{r_0^2}.
 \end{align}
Vectorizing the observation model from \eqref{eq:datacube}, the measurement model can be written as
\begin{align}
    \mathbf{y}_0
    =
    \boldsymbol{\mu}_0(\boldsymbol{\omega}_0)
    + \mathbf{r}_0
\end{align}
with $\{\mathbf{y}_0, \boldsymbol{\mu}_0(\boldsymbol{\omega}_0), \mathbf{r}_0\} = \mathrm{vec}(\{\mathbf{Y}_0, \mathbf{M}_0(\boldsymbol{\omega}_0), \mathbf{R}_0\})$, where $\mathrm{vec}(\cdot)$ denotes the vectorization operator.

\subsection{Information and Uncertainty Model}
The capability of the static agent to estimate the local parameter vector $\boldsymbol{\omega}_0$ from the measurement $\mathbf{y}_0$ is determined by the quality of the measurement, which is affected by the radar parameters and the object parameters. To evaluate the estimation performance, we derive the minimal achievable variance of any unbiased estimator of $\boldsymbol{\omega}_0$ based on the measurement $\mathbf{y}_0$. This is given by the CRLB, which is derived from the FIM about $\boldsymbol{\omega}_0$ \cite{scharf1991}, as
\begin{align}
    \text{CRLB}(\boldsymbol{\omega}_0) &= \mathbf{I}_0(\boldsymbol{\omega}_0)^{-1}.
\end{align}
In particular, the FIM for the local parameter vector for the proposed model is given by the following lemma.

\begin{Lemma}
\label{Lemma:FIM}
For the FMCW point-target model in \eqref{eq:datacube}, assume white complex Gaussian noise with covariance $\mathbf{\Sigma}_{R,0}=\sigma_{R,0}^2\mathbf{I}$. Then the FIM with respect to $\boldsymbol{\omega}_0=[r_0,\theta_0,v_0]^\top$ is diagonal, and its diagonal entries are approximately given by
\begin{align}
\label{eq:FIM_1}
[\mathbf{I}_0(\boldsymbol{\omega}_0)]_{1,1}
&= \frac{P_0 N M L}{6 \sigma_{R,0}^2}
\left(\frac{4 \pi S T_s}{c}\right)^2 (N^2-1), \\
\label{eq:FIM_2}
[\mathbf{I}_0(\boldsymbol{\omega}_0)]_{2,2}
&= \frac{P_0 N M L}{6 \sigma_{R,0}^2}
\left(\frac{2 \pi d \cos\theta_0}{\lambda}\right)^2 (M^2-1), \\
[\mathbf{I}_0(\boldsymbol{\omega}_0)]_{3,3}
\label{eq:FIM_3}
&= \frac{P_0 N M L}{6 \sigma_{R,0}^2}
\left(\frac{4 \pi T_c}{\lambda}\right)^2 (L^2-1),
\end{align}
where $S$, $T_s$, and $T_c$ are the chirp slope, the fast-time sampling period, and the chirp repetition period, respectively.
\end{Lemma}
\begin{proof}
Starting from the vectorized observation model, the FIM is obtained from the standard complex Gaussian expression \cite{scharf1991}
\vspace*{-0.4cm}
\begin{align}
    \label{eq:FIM_def}
    \mathbf{I}_0(\boldsymbol{\omega}_0)
    =
    2\,\mathrm{Re}\!\left\{
    \left(\frac{\partial \boldsymbol{\mu}_0}{\partial \boldsymbol{\omega}_0}\right)^H
    \mathbf{\Sigma}_{R,0}^{-1}
    \left(\frac{\partial \boldsymbol{\mu}_0}{\partial \boldsymbol{\omega}_0}\right)
    \right\}.
\end{align}
Evaluating the derivatives of $\boldsymbol{\mu}_0$ with respect to the normalized range, angle, and Doppler frequencies yields terms proportional to the fast-time, spatial, and slow-time indices, respectively. Under the assumption of white noise, i.e., $\mathbf{\Sigma}_{R,0}=\sigma_{R,0}^2\mathbf{I}$, and for centered indices $n, m, l$, the mixed sums vanish, so that the FIM becomes diagonal. Using the identity $\sum_{q=-(Q-1)/2}^{(Q-1)/2} q^2 = \frac{Q(Q^2-1)}{12}$,
closed-form expressions for the diagonal entries are obtained in the normalized parameter domain. Finally, applying the Jacobian of the transformation from the normalized frequencies to the physical parameters $(r_0,\theta_0,v_0)$ yields the stated expressions.

If the object reflectivity is modeled as an unknown nuisance parameter, its effect can be incorporated using the Schur complement. In the approximation adopted here, this contribution is neglected, since the dominant information is carried by the phase terms of the FMCW signal.
\end{proof}

Note that due to the diagonal structure of the FIM the estimation of the range, angle, and radial velocity are approximately decoupled. Moreover, the local estimation accuracy is inherently limited by the scene and radar parameters, as shown in \eqref{eq:FIM_1}-\eqref{eq:FIM_3}. The estimation uncertainty can be reduced by the fusion of multiple measurements, provided by the surrounding distributed static radar systems, which however introduces coupling among different estimated parameters in the aggregated FIM, as will be discussed in the following section.

\section{Ideal Cooperative Fusion of Distributed Radar Measurements}
To complement the local measurement of the agent $k=0$ with the measurements of a set of $K$ surrounding static radar systems, we will first introduce a common reference frame for the object parameters, and then derive the fused measurement model and the resulting CRLB for the global object state estimate. 
\subsection{Global Coordinate System}
While the local parameter vector $\boldsymbol{\omega}_0$ provides information about the object's state in the local coordinate system of the moving agent, the global parameter vector $\mathbf{\Omega} = \left[x,y,v_x,v_y\right]^T$ represents the object's state in a common global coordinate system, where $x$ and $y$ denote the position and $v_x$ and $v_y$ denote the velocity components of the object. 

Representing the measurement model with respect to the object state in a common global coordinate system requires the transformation $g_k: \mathbf{\Omega} \mapsto \boldsymbol{\omega}_k$.
It is a nonlinear set of functions that depend on the relative positions and orientations of the different radar systems relative to the object and is given by 
\begin{align}
    \boldsymbol{\omega}_k = g_k(\mathbf{\Omega}) = 
    \begin{bmatrix}
    \sqrt{(x-x_k)^2 + (y-y_k)^2} \\
    \arctan\left(\frac{y-y_k}{x-x_k}\right) - \phi_k \\
    \frac{(x-x_k)(v_x-v_{x,k}) + (y-y_k)(v_y-v_{y,k})}{\sqrt{(x-x_k)^2 + (y-y_k)^2}}
    \end{bmatrix},
\end{align}
where $(x_k,y_k)$ and $(v_{x,k},v_{y,k})$ denote the position and velocity of radar $k$, respectively, and $\phi_k$ denotes the orientation of radar $k$ relative to the global coordinate system.

\subsection{Fused Measurement Model}
The global measurement model from the moving agent $k=0$ and the $K$ supporting distributed radar systems can be expressed as 
\begin{align}
    \mathbf{y}(\boldsymbol{\Omega}) &= \begin{bmatrix}
    \boldsymbol{\mu}_0(g_0(\mathbf{\Omega})) \\
    \boldsymbol{\mu}_1(g_1(\mathbf{\Omega})) \\
    \vdots \\
    \boldsymbol{\mu}_K(g_K(\mathbf{\Omega}))
    \end{bmatrix} + 
    \begin{bmatrix}
    \mathbf{r}_0 \\
    \mathbf{r}_1 + \mathbf{d}_1 \\
    \vdots \\
    \mathbf{r}_K + \mathbf{d}_K
    \end{bmatrix}.
\end{align}
The additional noise terms $\mathbf{d}_k$ for $k=1,\dots,K$ capture the distortion induced by the quantization of the measurements of the supporting radar systems, which are shared with the moving agent. The distortion noise is assumed to be zero mean complex Gaussian distributed with covariance matrix $\mathbf{\Sigma}_{D,k}$, and uncorrelated with the sensor noise $\mathbf{r}_k \sim \mathcal{CN}(0, \mathbf{\Sigma}_{R,k})$. In the ideal setting, with unlimited communication resources, the distortion noise $\mathbf{d}_k$ can be neglected.

\begin{figure*}[!t]
\vspace{0cm}% vertical space for above-equation annotations
\begin{align}
    \label{eq:CRLB_tradeoff}
    \text{CRLB}(\mathbf{\Omega}_{t+\tau}) = \left(
    \underbrace{\tikzmarknode{nodeJ0a}{\mathbf{J}}_0^T \mathbf{I}\!\left(g_0\!\left(\boldsymbol{\Omega}_{t+\tau}\right)\right) \tikzmarknode{nodeJ0b}{\mathbf{J}}_0}_{\text{Local information}}
    + \underbrace{\sum_{k=1}^K \left(
    \mathbf{F}_\tau\,\left(
    \tikzmarknode{nodeJ}{\mathbf{J}}_k^T \mathbf{I}_k\!\left(g_k\!\left(\mathbf{\Omega}_t\right),\,
    \sigma_{R,k}^2 + \tikzmarknode{nodeDistNoise}{\sigma}_{D,k}^2
    \right)\tikzmarknode{nodeJ2}{\mathbf{J}}_k\right)^{-1} \mathbf{F}_\tau^T
    + \tikzmarknode{nodeQ}{\mathbf{Q}}_\tau
    \right)^{-1}}_{\text{Distributed radar information}} \right)^{-1}
\end{align}
% Compile twice for tikzmark positions to resolve correctly
\begin{tikzpicture}[remember picture, overlay, >=Stealth, shorten >=1mm,
    ann/.style={font=\footnotesize, align=center, inner sep=2pt}]
    % === "Impact of different sensor perspective": bar spanning all four J nodes ===
    \coordinate (midAllJ) at ($(nodeJ0a.north)!0.2!(nodeJ2.north)$);
    \coordinate (barY)    at ([yshift=0.5cm]midAllJ);
    \coordinate (forkL)   at (nodeJ0a.north |- barY);
    \coordinate (forkJ0b) at (nodeJ0b.north |- barY);
    \coordinate (forkJ)   at (nodeJ.north   |- barY);
    \coordinate (forkR)   at (nodeJ2.north  |- barY);
    \draw[shorten >=0mm] (forkL) -- (forkR);               % horizontal bar spanning all J nodes
    \draw[->] (forkL)   -- (nodeJ0a.north); % drop to J_0^T
    \draw[->] (forkJ0b) -- (nodeJ0b.north); % drop to J_0
    \draw[->] (forkJ)   -- (nodeJ.north);   % drop to J_k^T
    \draw[->] (forkR)   -- (nodeJ2.north);  % drop to J_k
    \node[ann, above] (labJ) at (barY) {Impact of different\\sensor perspective};
    % === "Impact of quantization" – shifted left to avoid overlap with AoI ===
    \node[ann, above] (labNoise)
        at ([xshift=0.0cm, yshift=0.75cm]nodeDistNoise.north)
        {Impact of\\quantization};
    \draw[->, dashed] (labNoise.south) -- (nodeDistNoise.north);
    % === "Age of Information" – shifted right to avoid overlap with quantization ===
    \node[ann, above] (labQ)
        at ([xshift=-0.4cm, yshift=0.65cm]nodeQ.north)
        {Impact of\\ Age of Information};
    \draw[->, dashed] ([xshift=0.4cm, yshift=0.0cm]labQ.south) -- (nodeQ.north);
\end{tikzpicture}
\hrulefill
\end{figure*}

\subsection{Global Information and Uncertainty}
In the same way as the object parameters are transformed from the local to the global parameter space, the uncertainty of the object estimates is transformed and propagated through the fusion process. The FIM about the global parameter vector $\mathbf{\Omega}$ can be derived from the FIMs about the local parameter vectors $\boldsymbol{\omega}_k$ by applying the chain rule for parameter transformation. The corresponding CRLB for the global parameter vector $\mathbf{\Omega}$ can then be obtained by inverting the global FIM, as given by Lemma \ref{Lemma:FusedCRLB}.
\begin{Lemma}
    \label{Lemma:FusedCRLB}
    Assuming the noise components are independent across agents, the CRLB for the global parameter vector $\mathbf{\Omega}$ can be expressed as the inverse of the sum of the projected FIMs from all agents as
    \begin{align}
        \text{CRLB}(\mathbf{\Omega}) = \left( \sum_{k=0}^K \mathbf{J}_k^T \mathbf{I}_k(g_k(\boldsymbol{\Omega})) \mathbf{J}_k \right)^{-1},
    \end{align}
    where $\mathbf{J}_k = \frac{\partial \boldsymbol{\omega}_k}{\partial \mathbf{\Omega}}$ is the Jacobian of the transformation from the global to the local parameter space, evaluated at the true parameter value, and $\mathbf{I}_k(\cdot)$ is given in \eqref{eq:FIM_def}.
\end{Lemma}

\begin{proof}
    Under the independence assumption, the individual measurements from the $K$ agents are conditionally independent given the global object state $\mathbf{\Omega}$. Consequently, the joint log-likelihood function is the sum of the marginal log-likelihoods of the individual agents. Hence, the global FIM is the sum of the individual FIMs. 
    Furthermore, the local parameter vector $\boldsymbol{\omega}_k$ is a deterministic function of the global state $\mathbf{\Omega}$. Using the chain rule for Fisher information under parameter transformation, the FIM of agent $k$ with respect to $\mathbf{\Omega}$ is derived by projecting its local FIM through the Jacobian matrix $\mathbf{J}_k = \frac{\partial \boldsymbol{\omega}_k}{\partial \mathbf{\Omega}}$, resulting in the quadratic form $\mathbf{J}_k^T \mathbf{I}(\boldsymbol{\omega}_k) \mathbf{J}_k$. Substituting this and $\boldsymbol{\omega}_k = g_k(\mathbf{\Omega})$ into the sum of the individual FIMs yields the stated result.
\end{proof}

\section{Age of Information and Distortion in Cooperative Fusion}

Cooperative fusion requires the agents to exchange their local measurements over capacity-limited wireless links. Let $C_{0,k}$ denote the channel capacity between radar $k$ and the moving agent, and let $a_k \in [0,1]$ be the fraction of resources allocated to radar $k$ under a multiple access scheme. 
The resulting rate constraint requires $a_k C_{0,k} \geq \frac{B_k}{\tau}$, where $B_k$ is the number of bits encoding the measurement of radar $k$ and $\tau$ is the update period. To satisfy this constraint, either $B_k$ can be reduced through quantization or $\tau$ can be increased, both of which degrade the shared measurement quality and reduce estimation performance. Note that $\tau$ sets the bit budget describing an already-acquired frame, not the intrinsic measurement resolution, which is fixed by the waveform.

\subsection{Quantization}

To quantify the distortion introduced by quantization, we use rate-distortion theory for Gaussian sources. Since the object parameters are unknown and no specific prior distribution is assumed, the signal $\mathbf{y}_k$ can be conservatively approximated as a zero mean complex Gaussian vector from a source coding perspective. The sensor noise $\mathbf{r}_k$ is also Gaussian, so that the overall source to be encoded is a Gaussian mixture. 
Rate-distortion theory \cite{Cover2001} then provides an upper bound on the distortion at any given rate.
With identical per-sample power across the fast-time, spatial, and slow-time dimensions, the covariance matrices therefore reduce to $\mathbf{\Sigma}_{Y,k} = \sigma_{Y,k}^2 \mathbf{I}$, $\mathbf{\Sigma}_{R,k} = \sigma_{R,k}^2 \mathbf{I}$, and $\mathbf{\Sigma}_{D,k} = \sigma_{D,k}^2 \mathbf{I}$. The $NML$ scalar components of $\mathbf{y}_k$ can thus be encoded independently. By the Gaussian rate-distortion function \cite{Cover2001}, encoding each component at $b_k = B_k/(NML)$ bits yields the following minimum achievable distortion per sample:
\begin{align}
    \label{eq:distortion}
    b_k = \log_2 \left(1+ \frac{\sigma_{Y,k}^2 + \sigma_{R,k}^2}{\sigma_{D,k}^2}\right) \Rightarrow \sigma_{D,k}^2 = \frac{\sigma_{Y,k}^2 + \sigma_{R,k}^2}{2^{b_k} - 1},
\end{align}
where $B_k$ is the total number of bits allocated to the measurement of radar~$k$. Since the distortion noise is independent, the noise covariance $\sigma^2\mathbf{I}$ in Lemma~\ref{Lemma:FIM} is replaced by $(\sigma_{R,k}^2 + \sigma_{D,k}^2)\mathbf{I}$.

\subsection{Latency and Information Aging}
Due to the update period $\tau$, the received measurement from radar $k$ reflects a past object state rather than the current one. Let $\mathbf{\Omega}_t$ denote the object state at the time the measurement was taken and $\mathbf{\Omega}_{t+\tau}$ the current state. We relate the two through a linear temporal evolution model
\begin{align}
    \mathbf{\Omega}_{t+\tau} = \mathbf{F}_\tau\, \mathbf{\Omega}_t + \mathbf{w}_\tau,\label{eq:StateTransModel}
\end{align}
where $\mathbf{F}_\tau$ is the state transition matrix and $\mathbf{w}_\tau$ is zero-mean process noise with covariance $\mathbf{Q}_\tau = \mathbb{E}[\mathbf{w}_\tau \mathbf{w}_\tau^T]$, uncorrelated with $\mathbf{\Omega}_t$. 

\begin{Lemma}
    Given the FIM $\mathbf{I}_k(\mathbf{\Omega}_t)$ obtained from the delayed measurement, the FIM about the current state $\mathbf{\Omega}_{t+\tau}$, defined according to \eqref{eq:StateTransModel}, is given by
    \begin{align}
        \mathbf{I}_k(\mathbf{\Omega}_{t+\tau}) = \left(\mathbf{F}_\tau\, \mathbf{I}_k(\mathbf{\Omega}_t)^{-1} \mathbf{F}_\tau^T + \mathbf{Q}_\tau\right)^{-1}.
    \end{align}
\end{Lemma}

\begin{proof}
    Under the linear Gaussian model in \eqref{eq:StateTransModel}, the covariance of the current state propagates as $\mathbf{P}_{t+\tau} = \mathbf{F}_\tau\, \mathbf{P}_t\, \mathbf{F}_\tau^T + \mathbf{Q}_\tau$. Substituting $\mathbf{P}_t = \mathbf{I}_k(\mathbf{\Omega}_t)^{-1}$ and inverting yields the result.
\end{proof}
The formulation of the state transition model $\mathbf{F}_\tau$ determines how the uncertainty in the past states propagates to the current state over the update period $\tau$. For an identity state transition model $\mathbf{F}_\tau = \mathbf{I}$, the uncertainties for the individual object parameters grow independently and linearly with $\tau$ due to the process noise $\mathbf{Q}_\tau$.
In contrast, for a constant-velocity model, the state transition model $\mathbf{F}_\tau$ correlates the position and velocity components of the object state to predict the future state. In that case, the uncertainty in the velocity components also contributes to the uncertainty in the position components, which causes the position uncertainty to grow \textit{cubically} with~$\tau$.

\subsection{Latency and Distortion Tradeoff}

Combining the effects of quantization and latency, the resulting CRLB for the global parameter vector $\mathbf{\Omega}_{t+\tau}$ is given by \eqref{eq:CRLB_tradeoff} at the top of this page, where the first part of the expression denotes the moving agent's own contribution and
the later part captures the contribution of the supporting radar systems. The notation $\mathbf{I}_k(\cdot\,;\, \sigma_{R,k}^2 + \sigma_{D,k}^2)$ denotes the local FIM evaluated with the combined sensor and distortion noise power $\sigma_{R,k}^2 + \sigma_{D,k}^2$ in place of solely $\sigma_{R,k}^2$, since the distortion $\mathbf{d}_k$ adds to the sensor noise $\mathbf{r}_k$.
%Despite the complex expression form due to the nonlinear variable transformations and the repeated matrix inversions, the CRLB in \eqref{eq:CRLB_tradeoff} provides a systematic way to gain insights into the tradeoff between latency and distortion in cooperative fusion. 
The CRLB in \eqref{eq:CRLB_tradeoff} provides a systematic way to gain insights into the tradeoff between latency and distortion.

In the following we will introduce some simplifications to further analyze the structure of the tradeoff. 
\textit{1)} We assume being interested only in the marginalized CRLB for the position components of the object state. Following from that, we approximate the latency induced uncertainty by the position variance, which grows cubically with $\tau$ as argued in the previous section. This cubic growth is approximated by $\tau^3 Q$, where $Q$ is the process noise intensity.
\textit{2)} In an attempt to simplify the nonlinear geometric relationships that cause the addition of structurally different FIMs from different radar systems, we assume that the information contribution from different perspectives are captured by the scalar $\Bar{I}_k$.
\textit{3)} We consider only the local information $\Bar{I}_0$ and one additional radar providing $\Bar{I}_1$. Substituting \eqref{eq:distortion} and applying the rate constraint $B_1 = \tau C_{0,1}$ yields
\vspace*{-0.4cm}
\begin{align}
    \label{eq:simplified}
    \text{CRLB}(\tau) \approx \biggl( \frac{\Bar{I}_0}{\sigma_{R,0}^2} + \biggl( \frac{\sigma_{R,1}^2 + \overbrace{\frac{\sigma_{Y,1}^2+\sigma_{R,1}^2}{2^{\tau C_{0,1}}-1}}^{\text{Distortion}}}{\Bar{I}_1} + \overbrace{\tau^3 Q}^{\text{AoI}}\biggr)^{-1} \biggr)^{-1}.
\end{align}
Here, it becomes apparent that for small update periods $\tau$, the CRLB is dominated by the distortion noise, while in the regime of large update periods, the CRLB is dominated by the latency-induced uncertainty. The optimal update period minimizing the CRLB can be found by numerically optimizing the expression.
\vspace*{-0.1cm}
\section{Simulation Results}
In this section, we present a numerical evaluation of the derived CRLB for a scenario depicted in Fig.~\ref{fig:posonly_setup}. The agent is moving along the trajectory shown in blue, while the object of interest is located at the origin. The agent is supported by two static radar systems, depicted in orange. The radar of the agent as well as the static radar systems operate with the same parameters: $N=1048$, $M=8$, $L=32$, $B=1\,\text{GHz}$, $T_c=50\,\mu\text{s}$, $\text{Transmit power} = 20\,\text{dBm}$, and $\text{Noise Figure} = 12\,\text{dB}$. The object has an RCS of $\sigma_\text{RCS} = 0.5\,\text{m}^2$. As shown in Fig.~\ref{fig:posonly_inclusion}, the CRLB is evaluated over time as the agent moves along the trajectory depicted in Fig.~\ref{fig:posonly_setup}. The estimation performance first improves as the agent approaches the object, but rapidly degrades once the object moves out of its field of view, at which point the supporting radar systems become the sole source of information. While the agent-only CRLB becomes unbounded at around 12\,$s$, including one or both static radars maintains a reliable fallback performance. 

In Fig.~\ref{fig:posonly_distortion}, we evaluate the effect of quantization distortion on the performance of the fused estimation case. The CRLB increases when fewer bits are allocated to the quantization of the measurements of the supporting radar systems. It can also be observed that the estimation performance gain decreases with the number of bits, which is not only a consequence of the logarithmic relationship between the distortion noise power and the number of bits, but also a consequence of the diminishing returns of the information gain from the supporting radar systems as the local measurement becomes relatively more informative. 

In Fig.~\ref{fig:posonly_update}, we evaluate the effect of latency on the performance of the fused estimation case. The CRLB increases within the update period $\tau$ due to the aging of the information. The plot shows that the degradation starts cubically and then approaches an upper limit posed by the CRLB from the local measurement of the agent. With decreasing update period the CRLB approaches the ideal case with no latency. 

The tradeoff between distortion noise and latency-induced uncertainty is illustrated in Fig.~\ref{fig:crlb_tradeoff}, where the normalized CRLB is plotted as a function of $\tau$ as given by the simplified expression from \eqref{eq:simplified}. The plot shows that the tradeoff leads to an optimal update period that minimizes the CRLB. However, finding the optimum in the general case with the full expression from \eqref{eq:CRLB_tradeoff} is highly complex and needs further investigation.

% Auto-generated by Simulations/position_only/passby_position_only.py
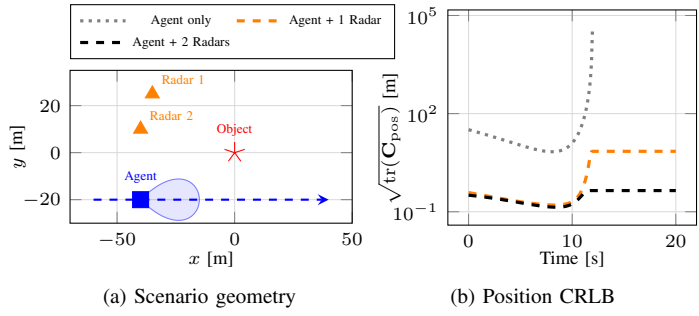
\begin{figure}[t]
\centering
\begin{minipage}[b]{0.5\columnwidth}
  \centering
  \hspace*{16pt}{\tiny\pgfplotslegendfromname{inclLegend}}\\[4pt]
  \subfloat[Scenario geometry]{\label{fig:posonly_setup}%
  \hspace*{-10pt}\begin{tikzpicture}
  \begin{axis}[
      width=1.2\columnwidth,
      xlabel={$x$~[m]}, ylabel={$y$~[m]},
      label style={font=\scriptsize},
      tick label style={font=\scriptsize},
      xlabel shift=-5pt, ylabel shift=-8pt,
      axis equal image, grid=both,
      grid style={line width=0.1pt, draw=gray!30},
      xmin=-70, xmax=50,
      ymin=-30, ymax=35,
  ]
  \draw[blue, dashed, thick, ->, >=stealth] (axis cs:-60,-20) -- (axis cs:40,-20);
  \fill[blue, opacity=0.1] (axis cs:-40,-20) .. controls (axis cs:-25,-5) and (axis cs:-15,-12) .. (axis cs:-15,-20) .. controls (axis cs:-15,-28) and (axis cs:-25,-35) .. cycle;
  \draw[blue, opacity=0.5] (axis cs:-40,-20) .. controls (axis cs:-25,-5) and (axis cs:-15,-12) .. (axis cs:-15,-20) .. controls (axis cs:-15,-28) and (axis cs:-25,-35) .. cycle;
  \addplot[only marks, mark=square*, mark size=3pt, blue] coordinates {(-40,-20)};
  \node[anchor=south, font=\tiny, blue] at (axis cs:-40,-17) {Agent};
  \addplot[only marks, mark=triangle*, mark size=3pt, orange] coordinates {(-35,25)};
  \node[anchor=south west, font=\tiny, orange] at (axis cs:-35,25) {Radar 1};
  \addplot[only marks, mark=triangle*, mark size=3pt, orange] coordinates {(-40,10)};
  \node[anchor=south west, font=\tiny, orange] at (axis cs:-40,10) {Radar 2};
  \addplot[only marks, mark=star, mark size=4pt, red] coordinates {(0,0)};
  \node[anchor=south, font=\tiny, red] at (axis cs:0,3) {Object};
  \end{axis}
  \end{tikzpicture}%
  }
\end{minipage}%
\hfill
\subfloat[Position CRLB]{\label{fig:posonly_inclusion}%
\begin{tikzpicture}
    \hspace{-0pt}
\begin{semilogyaxis}[
    width=0.55\columnwidth, height=0.5\columnwidth,
    xlabel={Time~[s]},
    ylabel={$\sqrt{\mathrm{tr}(\mathbf{C}_{\mathrm{pos}})}$~[m]},
    label style={font=\scriptsize},
    tick label style={font=\scriptsize},
    xlabel shift=-5pt, ylabel shift=-10pt,
    ylabel style={xshift=-8pt},
    grid=both,
    grid style={line width=0.1pt, draw=gray!30},
    legend to name=inclLegend,
    legend style={font=\tiny, legend columns=2},
    unbounded coords=jump,
]
\addplot[gray, dotted, very thick]
    table[x index=0, y index=1, col sep=comma, skip first n=1]
    {figures/data/posonly_inclusion.csv};
\addlegendentry{Agent only}

\addplot[orange, dashed, very thick]
    table[x index=0, y index=2, col sep=comma, skip first n=1]
    {figures/data/posonly_inclusion.csv};
\addlegendentry{Agent + 1 Radar}

\addplot[black, dashed, very thick]
    table[x index=0, y index=3, col sep=comma, skip first n=1]
    {figures/data/posonly_inclusion.csv};
\addlegendentry{Agent + 2 Radars}

\end{semilogyaxis}
\end{tikzpicture}%
}%
\caption{Scenario geometry (a) and position CRLB for different radar configurations (b).}
\label{fig:posonly_setup_inclusion}
\end{figure}

% Auto-generated by Simulations/position_only/passby_position_only.py
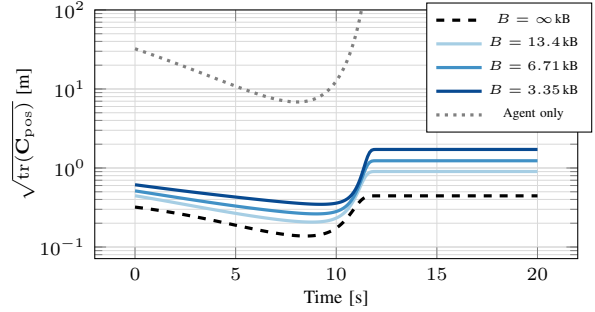
\begin{figure}[t]
\centering
\begin{tikzpicture}
\begin{semilogyaxis}[
    width=0.9\columnwidth,
    height=0.55\columnwidth,
    xlabel={Time~[s]},
    ylabel={$\sqrt{\mathrm{tr}(\mathbf{C}_{\mathrm{pos}})}$~[m]},
    label style={font=\scriptsize},
    tick label style={font=\scriptsize},
    xlabel shift=-3pt, ylabel shift=-5pt,
    grid=both,
    grid style={line width=0.1pt, draw=gray!30},
    legend pos=north east,
    legend style={font=\tiny, at={(1.03,1.03)}, anchor=north east},
    unbounded coords=jump,
    ymax=100.0,
]
\addplot[black, dashed, very thick]
    table[x index=0, y index=1, col sep=comma, skip first n=1]
    {figures/data/posonly_distortion.csv};
\addlegendentry{$B=\infty$\,kB}

\addplot[color={rgb,1:red,0.652;green,0.806;blue,0.894}, solid, very thick]
    table[x index=0, y index=2, col sep=comma, skip first n=1]
    {figures/data/posonly_distortion.csv};
\addlegendentry{$B=13.4$\,kB}

\addplot[color={rgb,1:red,0.256;green,0.570;blue,0.775}, solid, very thick]
    table[x index=0, y index=3, col sep=comma, skip first n=1]
    {figures/data/posonly_distortion.csv};
\addlegendentry{$B=6.71$\,kB}

\addplot[color={rgb,1:red,0.031;green,0.290;blue,0.570}, solid, very thick]
    table[x index=0, y index=4, col sep=comma, skip first n=1]
    {figures/data/posonly_distortion.csv};
\addlegendentry{$B=3.35$\,kB}

\addplot[gray, dotted, very thick]
    table[x index=0, y index=5, col sep=comma, skip first n=1]
    {figures/data/posonly_distortion.csv};
\addlegendentry{Agent only}

\end{semilogyaxis}
\end{tikzpicture}
\caption{Position CRLB for varying quantization rates $B$ [kB].}
\label{fig:posonly_distortion}
\end{figure}

% Auto-generated by Simulations/position_only/passby_position_only.py
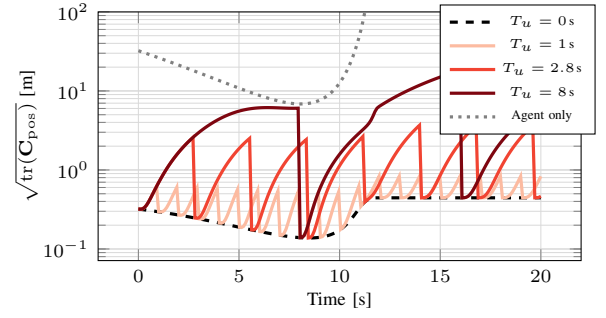
\begin{figure}[t]
\centering
\begin{tikzpicture}
\begin{semilogyaxis}[
    width=0.9\columnwidth,
    height=0.55\columnwidth,
    xlabel={Time~[s]},
    ylabel={$\sqrt{\mathrm{tr}(\mathbf{C}_{\mathrm{pos}})}$~[m]},
    label style={font=\scriptsize},
    tick label style={font=\scriptsize},
    xlabel shift=-3pt, ylabel shift=-5pt,
    grid=both,
    grid style={line width=0.1pt, draw=gray!30},
    legend pos=north east,
    legend style={font=\tiny, at={(1.03,1.03)}, anchor=north east},
    unbounded coords=jump,
    ymax=100.0,
]
\addplot[black, dashed, very thick]
    table[x index=0, y index=1, col sep=comma, skip first n=1]
    {figures/data/posonly_update.csv};
\addlegendentry{$T_u=0$\,s}

\addplot[color={rgb,1:red,0.988;green,0.732;blue,0.630}, solid, very thick]
    table[x index=0, y index=2, col sep=comma, skip first n=1]
    {figures/data/posonly_update.csv};
\addlegendentry{$T_u=1$\,s}

\addplot[color={rgb,1:red,0.947;green,0.268;blue,0.196}, solid, very thick]
    table[x index=0, y index=3, col sep=comma, skip first n=1]
    {figures/data/posonly_update.csv};
\addlegendentry{$T_u=2.8$\,s}

\addplot[color={rgb,1:red,0.495;green,0.022;blue,0.063}, solid, very thick]
    table[x index=0, y index=4, col sep=comma, skip first n=1]
    {figures/data/posonly_update.csv};
\addlegendentry{$T_u=8$\,s}

\addplot[gray, dotted, very thick]
    table[x index=0, y index=5, col sep=comma, skip first n=1]
    {figures/data/posonly_update.csv};
\addlegendentry{Agent only}

\end{semilogyaxis}
\end{tikzpicture}
\caption{Position CRLB for varying update periods $T_u$ [s].}
\label{fig:posonly_update}
\end{figure}

% CRLB tradeoff decomposition plot
\definecolor{darkblue}{rgb}{0.15,0.35,0.62}
\definecolor{darkred}{rgb}{0.65,0.10,0.15}
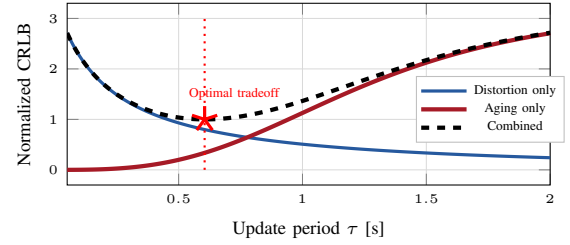
\begin{figure}[t]
\centering
\begin{tikzpicture}
\begin{axis}[
    width=0.9\columnwidth,
    height=0.45\columnwidth,
    xlabel={Update period $\tau$~[s]},
    ylabel={Normalized CRLB},
    grid=both,
    grid style={line width=0.1pt, draw=gray!30},
    label style={font=\scriptsize},
    tick label style={font=\tiny},
    legend style={font=\tiny, at={(1.03,0.6)}, anchor=north east, draw=gray!50, inner sep=1pt, row sep=-2pt},
    xmin=0.05, xmax=2.0,
]
\addplot[darkblue, solid, very thick]
    table[x=tau, y=crlb_distortion_only, col sep=comma]
    {figures/data/crlb_decomposition.csv};
\addlegendentry{Distortion only}

\addplot[darkred, solid, ultra thick]
    table[x=tau, y=crlb_aging_only, col sep=comma]
    {figures/data/crlb_decomposition.csv};
\addlegendentry{Aging only}

\addplot[black, dashed, ultra thick]
    table[x=tau, y=crlb_combined, col sep=comma]
    {figures/data/crlb_decomposition.csv};
\addlegendentry{Combined}

\addplot[red, dotted, thick] coordinates {(0.60491, 0) (0.60491, 3)};

\addplot[red, mark=star, mark size=5pt, only marks, very thick] coordinates {(0.60491, 1.0)};
\node[font=\tiny, anchor=south west] at (axis cs:0.60491-0.1, 1.2) {\color{red}Optimal tradeoff};

\end{axis}
\end{tikzpicture}
\caption{Normalized CRLB vs.\ update period $\tau$ along the rate constraint $B_k = \tau C_{0,k}$, decomposed into distortion-only and aging-only contributions alongside the combined curve.}
\label{fig:crlb_tradeoff}
\end{figure}

\section{Conclusion and Outlook}
In this paper, we have proposed a framework to analyze the performance of cooperative distributed radar sensing. We have derived the CRLB and investigated the influence of communication constraints, particularly latency and quantization, on the estimation performance. The results provide insights into the tradeoff between update period and bit allocation and can be used to design cooperative fusion strategies, subject to the assumptions of a single point target without clutter.

Future work will leverage the framework to design online sensor-selection and resource-allocation policies, providing a feasible set of operating points for a closed-loop control algorithm. In the longer term, the framework can serve as the uncertainty quantification component of a joint sensing-control problem, in which the agent's sensing, communication, and motion decisions are co-optimized subject to chance constraints, connecting naturally to the dual control problem, where inputs simultaneously regulate the system and probe its state.
\vspace*{-0.5cm}
\bibliographystyle{IEEEtran}
\bibliography{IEEEabrv,reference}

\end{document}